\documentclass[11pt]{article}
\usepackage{amssymb}
\usepackage{amsfonts}
\usepackage{bbm}
\usepackage{mathrsfs}
\usepackage{mathrsfs}
\usepackage{longtable,multirow,rotating}
\usepackage{supertabular}
\usepackage{multirow}
\usepackage{makecell}
\usepackage{threeparttable}
\usepackage{caption}
\usepackage{tikz}
\usetikzlibrary{arrows}

\makeatletter
\pgfdeclareshape{datastore}{
  \inheritsavedanchors[from=rectangle]
  \inheritanchorborder[from=rectangle]
  \inheritanchor[from=rectangle]{center}
  \inheritanchor[from=rectangle]{base}
  \inheritanchor[from=rectangle]{north}
  \inheritanchor[from=rectangle]{north east}
  \inheritanchor[from=rectangle]{east}
  \inheritanchor[from=rectangle]{south east}
  \inheritanchor[from=rectangle]{south}
  \inheritanchor[from=rectangle]{south west}
  \inheritanchor[from=rectangle]{west}
  \inheritanchor[from=rectangle]{north west}
  \backgroundpath{
    \southwest \pgf@xa=\pgf@x \pgf@ya=\pgf@y
    \northeast \pgf@xb=\pgf@x \pgf@yb=\pgf@y
    \pgfpathmoveto{\pgfpoint{\pgf@xa}{\pgf@ya}}
    \pgfpathlineto{\pgfpoint{\pgf@xb}{\pgf@ya}}
    \pgfpathmoveto{\pgfpoint{\pgf@xa}{\pgf@yb}}
    \pgfpathlineto{\pgfpoint{\pgf@xb}{\pgf@yb}}
 }
}
\makeatother

\usepackage{amsfonts,mathrsfs,latexsym,amsmath,amssymb,amsthm}

\newtheorem{theorem}{Theorem}

\newtheorem{lemma}{Lemma}

\newtheorem{definition}{Definition}

\newtheorem{remark}{Remark}

\begin{document}
\title{\bf{ New dynamic and verifiable multi-secret sharing schemes based on LFSR public key cryptosystem}}
\author{{\bf  Jing Yang$^{1*}$,  Fang-Wei Fu$^1$}\\
 {\footnotesize \emph{ 1. Chern Institute of Mathematics and LPMC, Nankai University}}\\
  {\footnotesize  \emph{Tianjin, 300071, P. R. China}}\\
  {\footnotesize  \emph{$^*$Corresponding author: yangjing0804@mail.nankai.edu.cn}}}
\date{}

\maketitle \noindent {\small {\bf Abstract} A verifiable multi-secret sharing (VMSS) scheme enables the dealer to share multiple secrets, and the deception of both participants and the dealer can be detected. After analyzing the security of VMSS schemes proposed by Mashhadi and Dehkordi in 2015, we illustrate that they cannot detect some deception of the dealer. By using nonhomogeneous linear recursion and LFSR public key cryptosystem, we introduce two new VMSS schemes. Our schemes can not only overcome the drawback mentioned above, but also have shorter private/public key length at the same safety level. Besides, our schemes have dynamism.
 \vskip 1mm

\noindent
 {\small {\bf Keywords:} Verifiable multi-secret sharing; Nonhomogeneous linear recursion; LFSR public key cryptosystem; Key length; Dynamism}

\vskip 3mm \baselineskip 0.2in
%%%%%%%%%%%%%%%%%%%%%%%%%%%%%%%%%%%%%%%%%%%%%%%%%%%%%%%%%%%%%%%%%%%%%%%%

%%%%%%%%%%%%%%%%%%%%%%%%%%%%%%%%%%%%%%%%%%%%%%%%%%%%%%%%%%%%%%%%%%%%%%%%%%
\section{Introduction}
\qquad With the rapid development of Internet, the secure storage and transmission of information have become more and more important. The security of the information depends on the security of the cryptosystem, which depends on the keys used in the system. It is obviously insecure to have only one key holder, therefore secret sharing was proposed to solve the problem by distributing the keys among several members, which is significant to not only prevent the overcentralization of the key management but also guarantee the integrity and confidentiality of the keys.

However, there are some problems in the initial secret sharing scheme:

(1) They can only share one secret once;

(2) They need secure channel to distribute shares;

(3) They cannot perceive the detective behavior of both the dealer and the participants;

(4) The shares held by participants cannot be reused;

(5) If participants join in or quit from the scheme, all the shares need to be updated;

(6) When the dealer changes the threshold, all the shares need to be altered.

In order to overcome the weakness of the original scheme, researchers have proposed several improved schemes in recent years. In 2004, Yang et al. presented a new multi-secret sharing scheme(YCH)\cite{Yang:2004:MSS:2626452.2627067}. Based on Feldman's scheme \cite{Feldman:1987:PSN:1382440.1383000}, Shao et al. proposed an improved scheme \cite{Shao:2005:NEV:2614701.2615126} in 2005, which still needs a private channel. In 2006, Zhao et al. proposed an effective VMSS scheme (ZZZ) \cite{Zhao:2007:PVM:1222223.1222354}. Since public key cryptography is utilized in the verification phase, the private channel is unnecessary.

In 2008, Massoud and Samaneh \cite{HADIANDEHKORDI20082262} presented two efficient VMSS schemes, which employ the intractability of the discrete logarithm and RSA cryptosystem \cite{Rivest:1978:MOD:359340.359342} to modify the YCH scheme. For simplicity, we call the first scheme in \cite{HADIANDEHKORDI20082262} MS1 scheme, and the second scheme in \cite{HADIANDEHKORDI20082262} MS2 scheme. In 2016, Liu et al. \cite{Liu:2016:AVM:2869973.2870260}  found that ZZZ scheme, MS1 scheme and MS2 scheme cannot resist cheating by the dealer, and proposed two modified schemes utilizing RSA encryption system. Similarly, we call the first scheme in \cite{Liu:2016:AVM:2869973.2870260} LZZ1 scheme, and the second scheme in \cite{Liu:2016:AVM:2869973.2870260} LZZ2 scheme. In 2015, Massoud and Samaneh proposed two new VMSS schemes (MS schemes) \cite{MASHHADI201531} by nonhomogeneous linear recursions and LFSR public key cryptosystem \cite{Gong:1999:PCB:2263211.2266021,10.1007/3-540-45537-X_22}. Likewise, the two schemes have the same drawback as the schemes in \cite{HADIANDEHKORDI20082262}, and we call the first scheme in \cite{MASHHADI201531} MS3 scheme, and the second scheme in \cite{MASHHADI201531} MS4 scheme.

In this work, we will present two new dynamic VMSS schemes using LFSR public key cryptosystem based on the MS schemes \cite{MASHHADI201531}, which overcome the disadvantages of the previous schemes and have shorter key length than the schemes in \cite{Liu:2016:AVM:2869973.2870260}. Moreover, our schemes allow participants to join in or quit from the group optionally and let the dealer to change the number or value of shared secrets, even the threshold according to practical situation dynamically.

The rest of this paper is organized as follows. In Section 2, we review the nonhomogeneous linear recursion, the LFSR public key cryptosystem, and give the attack to MS schemes. In Section 3, we present our two schemes. We propose the security analysis in Section 4, while Section 5 gives the performance analysis. Finally, we conclude our schemes in Section 6.

\section{Preliminaries}

\subsection{Linear recursion}
\qquad In this subsection, we introduce the linear recursion briefly, which you can refer to\cite{Biggs:2002:DM:579088} for a detailed description.

\begin{definition}%Definition 1

A linear recursion is defined by the equations:

$$
\begin{cases}
u_{0}=c_{0},u_{1}=c_{1},\cdots,u_{k-1}=c_{k-1},\\
u_{n+k}+a_{1}u_{n+k-1}+\cdots +a_{k}u_{n}=f(n) \quad (n\geq 0),
\end{cases}$$
where $c_{0},c_{1},\cdots, c_{k-1}$ and $a_{1}, a_{2}, \cdots,a_{k}$ are predefined real constants. $k$ is a positive variable, the degree of this linear recursion. If $f(n)=0$, the linear recursion is homogeneous. Otherwise, it is nonhomogeneous.
\end{definition}

%%%%%%%%%%%%%%%%%%%%%%%%%%%%%%%%%%%%%%%%%%%%%%%%%%%
\begin{definition}%Definition 2
For a linear sequence $u_{n}(n\geq 0)$ with dgree $k$ defined above, we give the following concepts:

(1) Auxiliary equation: $x^{k}+a_{1}x^{k-1}+\cdots+a_{k}=0$.

(2) Generating function: $U(x)=\Sigma_{n=0}^{\infty}u_{n}x^{n}$.
\end{definition}

\begin{lemma}%Lemma 1
We assume that sequence $u_{n}(n\geq 0)$ withe degree $k$, and its auxiliary equation is $(x-\alpha_{1})^{m_{1}}(x-\alpha_{2})^{m_{2}}\cdots(x-\alpha_{l})^{m_{l}}=0$, where $m_{1}+m_{2}+\cdots+m_{l}=k$.
Then its generating function is

$$ U(x)=\frac{R(x)}{(1-\alpha_{1}x)^{m_{1}}(1-\alpha_{2}x)^{m_{2}}\cdots(1-\alpha_{l}x)^{m_{l}}},$$
where $R(x)$ is a polynomial and $\deg(R(x))<k.$

And $u_{n}=p_{1}(n)\alpha_{1}^{n}+p_{2}(n)\alpha_{2}^{n}+\cdots+p_{l}(n)\alpha_{l}^{n},$
where $p_{j}(n)=A_{0}+A_{1}n+A_{2}n^{2}+\cdots+A_{m_{j-1}}n^{m_{j-1}} (j=1,2,\cdots,l)$ and $A_{0},A_{1},\cdots,A_{m_{j-1}}$ are constants defined by $a_{1},a_{2},\cdots,a_{k}.$
\end{lemma}
%%%%%%%%%%%%%%%%%%%%%%%%%%%%%%%%%%%%%%%%%%%%%%%%%%%%%%%%%%%%%%%%

Our schemes use two examples of nonhomogeneous linear recursion$(NLR)$ shown as follows:

\begin{theorem}%Theorem 1
Utilizing $[NLR1]$ to generate $NLR$ sequence $(u_{n})_{n\geq 0}$, where $[NLR1]$ have the following form:

\begin{equation}
[NLR1]=\left\{
\begin{aligned}
& u_{0}=c_{0},u_{1}=c_{1},\cdots,u_{k-1}=c_{k-1},\\
&\sum_{j=0}^{k}\left( {\begin{array}{*{20}{ccc}}
	k\\
	j
	\end{array}} \right)u_{n+k-j}=(-1)^{n}c \quad (n\geq 0),
\end{aligned}
\right.
\end{equation}
where $c,c_{0},c_{1},\cdots, c_{k-1}$  are predefined real constants. Therefore $u_{n}=p(n)(-1)^{n},$ where $p(n)$ is a polynomial of $n$ with degree at most $k$.
\end{theorem}

\begin{proof}
Utilizing equation (1), we obtain
\begin{align*}
\bigg(\sum_{j=0}^{k}\left( {\begin{array}{*{20}{ccc}}
	k\\
	j
	\end{array}} \right)x^{j}\bigg)\sum_{n=0}^{\infty}u_{n}x^{n}
    &=u_{0}+(u_{1}+ku_{0})x+\cdots+\bigg(\sum_{j=0}^{k-1}\left( {\begin{array}{*{20}{ccc}}
	k\\
	j
	\end{array}} \right)u_{k-1-j}\bigg)x^{k-1}\\
    &+\bigg(\sum_{j=0}^{k}\left( {\begin{array}{*{20}{ccc}}
	k\\
	j
	\end{array}} \right)u_{k-j}\bigg)x^{k}+\bigg(\sum_{j=0}^{k}\left( {\begin{array}{*{20}{ccc}}
	k\\
	j
	\end{array}} \right)u_{k+1-j}\bigg)x^{k+1}+\cdots\\
    &\overset{(1)}{=}u_{0}+(u_{1}+ku_{0})x+\cdots+\bigg(\sum_{j=0}^{k-1}\left( {\begin{array}{*{20}{ccc}}
	k\\
	j
	\end{array}} \right)u_{k-1-j}\bigg)x^{k-1}\\
    &+cx^{k}-cx^{k+1}+\cdots\\
\end{align*}
\begin{align*}
    &=g_{1}(x)+cx^{k}(1-x+x^{2}-x^{3}+\cdots)\\
    &=g_{1}(x)+\dfrac{cx^{k}}{1+x}\\
    &=\dfrac{g_{1}(x)(1+x)+cx^{k}}{1+x}
\end{align*}
where $g_{1}(x)$ is a polynomial with degree $(k-1)$. Consequently,
$$\sum_{n=0}^{\infty}u_{n}x^{n}=\dfrac{g_{1}(x)(1+x)+cx^{k}}{(1+x)^{k+1}}.$$
From Lemma 1, we can get $u_{n}=p(n)(-1)^{n}$, where $p(n)$ is a at most $k$-degree polynomial.
\end{proof}

%%%%%%%%%%%%%%%%%%%%%%%%%%%%%%%%%%%%%%%%%%%%%%%%%%%%%%%%%%%%%%%%%%%%%%
\begin{theorem}%Theorem 2
Utilizing $[NLR2]$ to generate $NLR$ sequence $(u_{n})_{n\geq 0}$, where $[NLR2]$ have the following form:

\begin{equation}
[NLR2]=\left\{
\begin{aligned}
& u_{0}=c_{0},u_{1}=c_{1},\cdots,u_{k-1}=c_{k-1},\\
&\sum_{j=0}^{k}\left( {\begin{array}{*{20}{ccc}}
	k\\
	j
	\end{array}} \right)(-1)^{j}u_{n+k-j}=c \quad (n\geq 0),
\end{aligned}
\right.
\end{equation}
where $c,c_{0},c_{1},\cdots, c_{k-1}$  are predefined real constants. Therefore $u_{n}=p(n)$, where $p(n)$ is a polynomial of $n$ with degree at most $k$.
\end{theorem}

The proof of Theorem 2 can be completed by the method analogous to Theorem 1.

%%%%%%%%%%%%%%%%%%%%%%%%%%%%%%%%%%%%%%%%%%%%%%%%%%%%%%%%%%%%%%%%%
\subsection{The LFSR public key cryptosystem}

\qquad At first, we introduce the third-order LFSR sequence \cite{Gong:1999:PCB:2263211.2266021,10.1007/3-540-45537-X_22}. Assuming that $f(x)=x^{3}-ax^{2}+bx-1$ is irreducible which is the characteristic polynomial of the following LFSR sequences, where $a,b\in F=GF(p)$ and $p$ is a prime.

\begin{definition}%Definition 3
A sequence $s=(s_{k})_{k\geq 0}$ satisfies the following conditions:
$$\begin{cases}
s_{0}=3,s_{1}=a,s_{2}=a^{2}-2b,\\
s_{k+3}-as_{k+2}+bs_{k+1}-s_{k}=0 \quad(k\geq 0).
\end{cases}$$
Then, we call $s=(s_{k})_{k\geq 0}$ is a third-order LFSR sequence whose characteristic polynomial is $f(x)$.
\end{definition}

We denote $s_{k}$ as $s_{k}(a,b)$, $s$ as $s(a,b)$, then we have the following Lemma.

\begin{lemma}\rm\cite{Gong:1999:PCB:2263211.2266021}\it%Lemma 2
Let $f(x)=x^{3}- ax^{2} + bx- 1$ over $F = GF(p)$ generate the three-order LFSR sequences $s=(s_{k})_{k\geq 0}$. If $s_{-k}(a,b) = s_{k}(b,a)$, then $s_{k}(s_{e}(a,b),s_{-e}(a,b))=s_{ke}(a,b)$ for all positive integer $k$ and $e$.

\end{lemma}

%%%%%%%%%%%%%%%%%%%%%%%%%%%%%%%%%%%%%%%%%%%%%%%%%%%%%%%%%%%%%%%%%%%%%%%%%%%

\begin{definition}%Definition 5
(The LFSR public key cryptosystem)
A sender generates the public key and private key by the following operation:

(1) selects two primes $p$ and $q$, and computes $N=pq$. Notice that the next few steps are performed on $Z_{N}$ and the period of the irreducible polynomial is $\delta = (p^{2} + p + 1)(q^{2} + q + 1)$;

(2) selects $e$ such that $gcd(e,p^{i}-1)$, where $i=2,3$;

(3) computes $d$ such that $de = 1\: mod\: \delta$;

(4) publishes $e$ and $N$ as public key, then keeps $d$ as private key.

\noindent{\textbf{Enciphering:}}

In order to send $m = (m_{1},m_{2})(0 < m_{1}, m_{2} < N)$ to the receiver, the sender generates $c = (c_{1}, c_{2})$ as corresponding cipher text, where $c_{1} = s_{e}(m_{1},m_{2}), c_{2} = s_{-e}(m_{1},m_{2})$.

\noindent{\textbf{Deciphering:}}

When receiving $c = (c_{1}, c_{2})$, the receiver can get the corresponding plain text $m_{1}= s_{d} (c_{1}, c_{2}), m_{2}= s_{-d}(c_{1}, c_{2})$ by private key $d$.

\end{definition}

%%%%%%%%%%%%%%%%%%%%%%%%%%%%%%%%%%%%%%%%%%%%%%%%%%%%%%%%%%%%%%%%%%%%%%%%
\subsection{Attack to MS schemes}

\qquad In this subsection, we give the attack to MS3 scheme, which is also true of MS4 scheme. Please refer to \cite{MASHHADI201531} for details of MS schemes. When recovering the secrets, they merely check the validity of $I_{i}'=s_{e_{i}}(s_{e_{0}}(a,b),s_{-e_{0}}(a,b))$ by $s_{d} (I_{i}',I_{-i}')=s_{e_{i}}(a,b)$, where $I_{-i}'=s_{-e_{i}}(s_{e_{0}}(a,b),s_{-e_{0}}(a,b))$, however, the consistency between $I_{i}'$ and $\{u_{i}\}$ is not checked. Thus, a malicious dealer can deceive the participants successfully, which means that:

When $1\leq i\leq k-1$,

(1) $D$ chooses a random $c<q_{1}$ and substitutes $R_{i}$ with $R_{i}'$ to calculate the equations below:
$$\left\{
\begin{aligned}
& u_{0}=R_{1},u_{1}=R_{2},\cdots,u_{i-1}'=R_{i}',\cdots,u_{k-2}=R_{k-1},\\
& \sum_{j=1}^{k}\left( {\begin{array}{*{20}{ccc}}
	k-1\\
	j-1
	\end{array}} \right)u_{n+k-j}=(-1)^{n}c\:mod \:q_{1} \quad (n\geq 0);
\end{aligned}
\right.$$
For $k-1 \leq i \leq m+l$, $D$ calculates $u_{i}$;

(2) $D$ calculates $y_{i}=R_{i}-u_{i-1}(k \leq i \leq m)$, and $r_{i} = S_{i}- u_{i+m}(1\leq i \leq l)$;

(3) $D$ releases $(s_{e_{0}}(a,b),d,r_{1},r_{2},\cdots,r_{l},y_{k},y_{k+1},\cdots,y_{m})$.

When $k-1< i\leq m$,

(1) $D$ chooses a random $c<q_{1}$ and considers the sequence generated by the equations below:
$$\left\{
\begin{aligned}
& u_{0}=R_{1},u_{1}=R_{2},\cdots,u_{k-2}=R_{k-1},\\
& \sum_{j=1}^{k}\left( {\begin{array}{*{20}{ccc}}
	k-1\\
	j-1
	\end{array}} \right)u_{n+k-j}=(-1)^{n}c\:mod \:q_{1} \quad (n\geq 0);
\end{aligned}
\right.$$
For $k-1 \leq i \leq m+l$, $D$ calculates $u_{i}$;

(2) $D$ substitutes $R_{i}$ with $R_{i}'$ to calculate $y_{i}'=R_{i}'-u_{i-1}$, then calculates $y_{j}(k-1 < j \leq m,j \neq i)$ and $r_{i} = P_{i} - u_{i+n}(1 \leq i \leq l)$ respectively;

(3) $D$ releases $(s_{e_{0}}(a,b),d,r_{1},r_{2},\cdots,r_{l},y_{k},y_{k+1},\cdots,y_{i}',\cdots,y_{m})$.

In the reconstruction phase, because the replacement is barely perceptible by $P_{i}$, those $k$ participants still provide real $R_{i}$ which conflicts with $\{u_{i}\}$ or $\{y_{i}\}$ generated by the dealer. Therefore, the recovered secrets are wrong. However, any at least $k$ honest participants exclusive of $M_{i}$ can reconstruct the shared secrets. Furthermore, if the dealer replaces more than one $R_{i}$ with invalid $R_{i}'$ , the situation gets even worse. In conclusion, the MS schemes cannot resist attack by a malicious dealer.
%%%%%%%%%%%%%%%%%%%%%%%%%%%%%%%%%%%%%%%%%%%%%%%%%%%%%%%%%%%%%%%%%%%%%%%%%%%
\section{The new VMSS schemes}

\qquad To avoid the attack mentioned above, based on MS schemes \cite{MASHHADI201531}, we present new VMSS schemes by examining consistency, which can detect deception of both participants and the dealer successfully.

\subsection{Scheme 1}
\qquad Scheme 1 utilizes the $[NLR1]$, the LFSR public key cryptosystem and the discrete logarithm problem.

\subsubsection{Initialization phase}

\qquad Suppose $D$ be the dealer, $P=\{P_{1},P_{2},\cdots,P_{m}\}$ be participants, and $k(k<m)$ be the threshold.

\noindent\textbf{Initialization of $D$:}

(1) $D$ selects $p_{0}$, $q_{0}$ $(p_{0} > q_{0})$ of bit-length $\lambda/2$, where $p_{0}$ and $q_{0}$ are two strong primes. Then $D$ calculates $N = p_{0}q_{0}$ of bit-length $\lambda$. Note that $\lambda$ here is the security parameter of the LFSR public key cryptosystem.

(2) $D$ randomly selects two primes $Q, q$ with bit-length more than $\lambda/2$, satisfying $Q|(q - 1)$, and
$Q>\left( {\begin{array}{*{20}{ccc}}
	k\\
	i
	\end{array}} \right)$
for $i=0,1,\cdots,k.$

(3) $D$ selects $g$ of $Z_{q}^{\ast}$ of order $Q$ satisfying that the discrete logarithm problem with base $g$ in $Z_{q}^{\ast}$ is infeasible.

(4) $D$ releases $(\lambda, N, Q, q, g)$ to participants.

\noindent\textbf{Initialization of participants:}

(1) $P_{i}$ of identity $ID_{i}$ selects two strong primes $p_{i}$ , $q_{i}$ $(p_{i} > q_{i})$, and then calculates $N_{i}= p_{i}q_{i}$ $(N_{i} > N)$. Note that the period of the irreducible polynomial in $Z_{N_{i}}$ is $\delta_{i}=(p_{i}^{2}+p_{i}+1)(q_{i}^{2}+q_{i}+1)$, then all the computations are performed in $Z_{N_{i}}$.

(2) $P_{i}$ randomly selects an integer $e_{i}\in [2,\delta_{i}]$ such that $gcd(e_{i},p_{i}^{j}-1)=1$, for $j=2,3$.

(3) $P_{i}$ calculates the integer $d_{i}$ satisfying $e_{i}d_{i}=1\:mod\: \delta_{i}$.

(4) $P_{i}$ passes $(ID_{i},e_{i},N_{i})$ to $D$ with a public channel, and keeps its shadow $d_{i}$ secret.

$D$ releases $(ID_{i},e_{i},N_{i}),i=1,2,\cdots,m$.

\begin{remark}%Remark 1
The released messages can be reused after this phase. In addition, $D$ cannot get any information about shadows, therefore these shadows are also reusable.
\end{remark}
%%%%%%%%%%%%%%%%%%%%%%%%%%%%%%%%%%%%%%%%%%%%%%%%%%%%%%%%%%%%%%%%%%%%%%%%%%%
\subsubsection{Construction phase}

\qquad Let $S_{1}, S_{2},\cdots, S_{l} \in Z_{Q}$ be $l$ secrets, where $0< S_{i} < Q(i=1,2,\cdots,l$). Then $D$ executes the steps as below to produce respective subshadow $u_{i-1}$:

(1) $D$ selects $c_{i} \in Z_{Q}^{\ast}$ for $i = 0,1,\cdots,k-1$ at random.

(2) $D$ randomly selects a constant $c$ satisfying $c<Q$ and considers $[NLR1]$ as below:
$$[NLR1]=\left\{
\begin{aligned}
& u_{0}=c_{0},u_{1}=c_{1},\cdots,u_{k-1}=c_{k-1},\\
& \sum_{j=0}^{k}\left( {\begin{array}{*{20}{ccc}}
	k\\
	j
	\end{array}} \right)u_{n_k-j}=(-1)^{n}c\:mod\:Q \quad (n\geq 0).
\end{aligned}
\right.$$

(3) For $ k \leq i \leq m+l $, $D$ calculates $u_{i}$.

(4) $D$ calculates $y_{i} = S_{i} - u_{m+i-1}\: mod\: Q(i = 1, 2,\cdots, l)$.

(5) $D$ calculates $H_{i}=s_{e_{i}}(u_{i-1},u_{i-1})\:mod\:N_{i}$ and $T_{i} = g^{u_{i-1}}\: mod\: q $,($1 \leq i \leq m$).

(6) $D$ releases $(H_{1}, H_{2},\cdots, H_{m}, T_{1}, T_{2},\cdots, T_{m},y_{1},y_{2},\cdots,y_{l},c,u_{m+l})$.

\begin{remark}%Remark 2
According to Lemma 2, we know that $s_{-k}(a,b)=s_{k}(b,a)$. If $b=a$, we have $s_{-k}(a,a)=s_{k}(a,a)$, which means that

$$H_{i}=s_{e_{i}}(u_{i-1},u_{i-1})\:mod\:N_{i}=s_{-e_{i}}(u_{i-1},u_{i-1})\:mod\:N_{i},\:1 \leq i \leq m.$$

\end{remark}
%%%%%%%%%%%%%%%%%%%%%%%%%%%%%%%%%%%%%%%%%%%%%%%%%%%%%%%%%%%%%%%%%%%%%%%%%%
\subsubsection{Verification phase}

\qquad In order to obtain the subshadow $u_{i-1}$, $P_{i}$ calculates following formula:
\begin{align*}
u_{i-1} & =s_{1}(u_{i-1},u_{i-1})\\
        & =s_{d_{i}}(s_{e_{i}}(u_{i-1},u_{i-1}),s_{-e_{i}}(u_{i-1},u_{i-1}))\\
        & =s_{d_{i}}(H_{i},H_{i})\:mod\:N_{i},1\leq i\leq m.
\end{align*}

By the formulas below, our schemes can perform validity and consistency detection.

$$ \prod_{j=0}^{k}(T_{i+1+t-j})^{\left( {\begin{array}{*{20}{ccc}}
	k\\
	j
	\end{array}} \right)}\overset{?}{=}g^{(-1)^{i}c}\:mod\:q$$

$$T_{i}\overset{?}{=}g^{u_{i-1}}\: mod\:q$$

Once the equations above are satisfied, each $u_{i-1}$ is thought to be valid and in accord with released messages. If every verification succeeds, participants think that $D$ is not malicious.
%%%%%%%%%%%%%%%%%%%%%%%%%%%%%%%%%%%%%%%%%%%%%%%%%%%%%%%%%%%%%%%%%%%%%%%%%

\subsubsection{Reconstruction phase}

\qquad Suppose that at least $k$ participants $\{P_{i}\}_{i\in I}(I\subseteq \{1,2,\cdots,m\})$ utilize corresponding $\{u_{i-1}\}_{i\in I}$ to reconstruct the secrets. Every $P_{i}$ can detect the validity of $\{u_{j-1}\}_{j\in I}(j\neq i)$ using the formulas as below:

$$g^{u_{j-1}}\overset{?}{=}T_{j}\:mod\:q,\quad j\in I,j\neq i.$$

The following two methods can be utilized:

\textbf{Method 1:} Using $k$ valid subshadows $\{u_{i-1}|i \in J \subseteq I, |J|=k\}$ and the released $u_{m+l}$, they can get the formulas by Theorem 1:

$$z_{1}+z_{2}(i-1)+\cdots+z_{k+1}(i-1)^{k}=u_{i-1}(-1)^{i-1}\:mod\:Q,i\in \{J \cup (m+l+1)\}.$$
Solving the equations or utilizing the Lagrange interpolation, they get $z_{1} = A_{1}, z_{2} =
A_{2},\cdots, z_{k+1} = A_{k+1}$ in $Z_{Q}$.

Then, they obtain
$$u_{i-1}=(A_{1}+A_{2}(i-1)+\cdots+A_{k+1}(i-1)^{k})(-1)^{i-1}\: mod\:Q,$$
where $i\in\{1,2,\cdots,m+l+1\}/\{J\cup (m+l+1)\}$.

Therefore, they reconstruct the secrets: $S_{i}=y_{i}+u_{m+i-1}\:mod\:Q\,(i=1,2,\cdots,l).$

\textbf{Method 2:} If utilizing $k$ successive $\{u_{i},u_{i+1},\cdots,u_{i+k-1}\}$, they can calculate other $u_{j},j\in\{0,1,\cdots,m+l\}/\{i,i+1,\cdots,i+k-1\}$ by the formulas:
$$\sum_{j=0}^{k}\left( {\begin{array}{*{20}{ccc}}
	k\\
	j
	\end{array}} \right)u_{n+k-j}=(-1)^{n}c\:mod\:Q \quad(n\geq 0).$$

Therefore, they reconstruct the secrets: $S_{i}=y_{i}+u_{m+i-1}\:mod\:Q\,(i=1,2,\cdots,l).$

%%%%%%%%%%%%%%%%%%%%%%%%%%%%%%%%%%%%%%%%%%%%%%%%%%%%%%%%%%%%%%%%%%%%%%%%%%%%%
\subsection{Scheme 2}

\qquad Scheme 2 utilizes the $[NLR2]$, the LFSR public key cryptosystem and the discrete logarithm problem.
\subsubsection{Initialization phase}
\qquad The initialization phase in Scheme 2 is the same as Scheme 1.
\subsubsection{Construction phase}
\qquad Compared with Scheme 1 ,we substitute $[NLR1]$ with the $[NLR2]$, and the rest is the same.

%%%%%%%%%%%%%%%%%%%%%%%%%%%%%%%%%%%%%%%%%%%%%%%%%%%%%%%%%%%%%%%%%%%%%%%%%%%%
\subsubsection{Verification phase}

\qquad $P_{i}$ can calculate $u_{i-1}=s_{d_{i}}(H_{i},H_{i})\:mod\:N_{i}(1\leq i\leq m)$ to obtain corresponding subshadow. By the formulas below, our schemes can perform validity and consistency detection.:

$$ \prod_{j=0}^{k}(T_{i+1+k-j})^{(-1)^{j}\left( {\begin{array}{*{20}{ccc}}
	k\\
	j
	\end{array}} \right)}\overset{?}{=}g^{c}\:mod\:q$$

$$T_{i}\overset{?}{=}g^{u_{i-1}}\: mod\:q$$

Once the equations above are satisfied, each $u_{i-1}$ is thought to be valid and in accord with released messages. If every verification succeeds, participants think that $D$ is not malicious.

%%%%%%%%%%%%%%%%%%%%%%%%%%%%%%%%%%%%%%%%%%%%%%%%%%%%%%%%%%%%%%%%%%%%%%%%%%%%%%%%
\subsubsection{Reconstruction phase}
\qquad Suppose that at least $k$ participants $\{P_{i}\}_{i\in I}(I\subseteq \{1,2,\cdots,m\})$ utilize corresponding $\{u_{i-1}\}_{i\in I}$ to reconstruct the secrets. Every $P_{i}$ can detect the validity of $\{u_{j-1}\}_{j\in I}(j\neq i)$ using the formulas as below:

$$g^{u_{j-1}}\overset{?}{=}T_{j}\:mod\:q,\quad j\in I,j\neq i.$$

\textbf{Method 1:} Using $k$ valid subshadows $\{u_{i-1}|i \in J \subseteq I, |J|=k\}$ and the released $u_{m+l}$, they can get the formulas by Theorem 2:

$$z_{1}+z_{2}(i-1)+\cdots+z_{k+1}(i-1)^{k}=u_{i-1}\:mod\:Q,i\in\{J\cup (m+l+1)\}.$$
Solving the equations or utilizing the Lagrange interpolation, they get $z_{1} = A_{1}, z_{2} =
A_{2},\cdots, z_{k+1} = A_{k+1}$ in $Z_{Q}$.

Then, they obtain
$$u_{i-1}=A_{1}+A_{2}(i-1)+\cdots+A_{k+1}(i-1)^{k}\: mod\:Q,$$
where $i\in\{1,2,\cdots,m+l+1\}/\{J\cup (m+l+1)\}$.

Therefore, they reconstruct the secrets: $S_{i}=y_{i}+u_{m+i-1}\:mod\:Q\,(i=1,2,\cdots,l).$

\textbf{Method 2:} If utilizing $k$ successive $\{u_{i},u_{i+1},\cdots,u_{i+k-1}\}$, they can calculate other $u_{j},j\in\{0,1,\cdots,m+l\}/\{i,i+1,\cdots,i+k-1\}$ by the formulas:

$$\sum_{j=0}^{k}\left( {\begin{array}{*{20}{ccc}}
	k\\
	j
	\end{array}} \right)(-1)^{j}u_{n+k-j}=c\:mod\:Q \quad (n\geq 0).$$

Therefore, they reconstruct the secrets: $S_{i}=y_{i}+u_{m+i-1}\:mod\:Q\,(i=1,2,\cdots,l).$

%%%%%%%%%%%%%%%%%%%%%%%%%%%%%%%%%%%%%%%%%%%%%%%%%%%%%%%%%%%%%%%%%%%%%%%%%%%%%

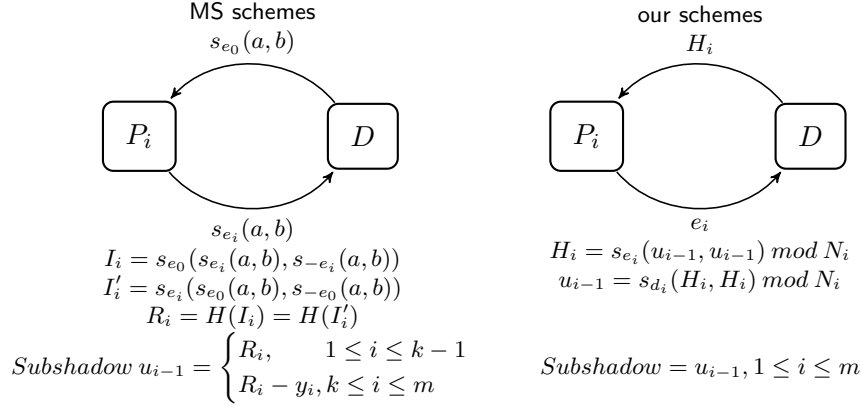
\begin{figure}[!htbp]
\begin{center}
\begin{tikzpicture}[
  font=\sffamily,
  every matrix/.style={ampersand replacement=\&,column sep=2cm,row sep=2cm},
  source/.style={draw,thick,rounded corners,inner sep=.3cm},
  process/.style={draw,thick,circle,fill=green!20},
  sink/.style={source},
  datastore/.style={draw,very thick,shape=datastore,inner sep=.3cm},
  dots/.style={gray,scale=2},
  to/.style={->,>=stealth',shorten >=1pt,semithick,font=\sffamily\footnotesize},
  every node/.style={align=center}]

  % Position the nodes using a matrix layout
  \matrix{

    \node[sink] (1) {$P_{i}$};
      \& \node[source] (2) {$D$};\qquad\qquad\qquad

       \& \node[sink] (3) {$P_{i}$};
       \& \node[source] (4) {$D$};
     \\
  };

  % Draw the arrows between the nodes and label them.
  \draw[to] (1) to[bend right=50]
      node[midway,below] {$s_{e_{i}}(a,b)$\\$I_{i}=s_{e_{0}}(s_{e_{i}}(a,b),s_{-e_{i}}(a,b))$\\
      $I_{i}'=s_{e_{i}}(s_{e_{0}}(a,b),s_{-e_{0}}(a,b))$\\$R_{i}=H(I_{i})=H(I_{i}')$\\
      $Subshadow\: u_{i-1}=\begin{cases}
                  R_{i} ,\qquad 1\leq i\leq k-1\\
                  R_{i}-y_{i},k\leq i\leq m
                 \end{cases}$} (2);
  \draw[to] (2) to[bend right=50] node[midway,above] {MS schemes\\$s_{e_{0}}(a,b)$}
      (1);
  \draw[to] (3) to[bend right=50]
      node[midway,below] {$e_{i}$\\$H_{i}=s_{e_{i}}(u_{i-1},u_{i-1})\:mod\:N_{i}$\\
      $u_{i-1}=s_{d_{i}}(H_{i},H_{i})\:mod\:N_{i}$\\\\\\$Subshadow=u_{i-1},1\leq i \leq m $} (4);
  \draw[to] (4) to[bend right=50] node[midway,above] {our schemes\\$H_{i}$}
      (3);

\end{tikzpicture}

\end{center}
\captionsetup{justification=centering}
\caption{The difference between MS schemes and our schemes}\label{fig:figname}
\end{figure}

%%%%%%%%%%%%%%%%%%%%%%%%%%%%%%%%%%%%%%%%%%%%%%%%%%%%%%%%%%%%%%%%%%%%%%%%%%%%%%%%

From the Figure 1, in MS schemes \cite{MASHHADI201531}, every $P_{i}$ selects $e_{i}$ independently, and computes corresponding $s_{e_{i}}(a,b)$, then transmits it to $D$. Thereafter, $D$ and $P_{i}$ can obtain secret share $R_{i}=H(I_{i})=H(I_{i}')$, where $I_{i}=s_{e_{0}}(s_{e_{i}}(a,b),s_{-e_{i}}(a,b)),I_{i}'=s_{e_{i}}(s_{e_{0}}(a,b),s_{-e_{0}}(a,b))$.
However, whether the $R_{i}$ used in the generation of $\{u_{i-1}\}$ is the same as the one provided by $P_{i}$ is not verified.

While in our schemes, $P_{i}$ computes $d_{i}$ and maintains confidentiality of $D$, then transmits $e_{i}$ to $D$, where $e_{i}$ is relevant public key of $d_{i}$. Then $D$ can get $H_{i}$ by $e_{i}$, where $H_{i}=s_{e_{i}}(u_{i-1},u_{i-1})\:mod\:N_{i}$. After that $P_{i}$ can calculate its subshadow $u_{i-1}$ by $d_{i}$. We add the consistency detection between the $u_{i-1}$ and released messages to perceive malicious dealer.

%%%%%%%%%%%%%%%%%%%%%%%%%%%%%%%%%%%%%%%%%%%%%%%%%%%%%%%%%%%%%%%%%%%%%%%%%%%%%

\section{Security analysis}
\qquad The security analysis of our two schemes is analogous, so we take Scheme 1 as example. Generally, when it comes to the security of a $(k,m)$ VMSS scheme, there are three conditions to be satisfied.

\textbf{(1)\,Correctness}: Provided that all participants and their dealer behave authentically, at least $k$ participants can recover shared secrets successfully.

\textbf{(2)\,Verifiability}:

$\bullet$ In the verification phase, any participant can detect dishonest operation by the dealer.

$\bullet$ In the reconstruction phase, any other participants can detect a false subshadow by a malicious participant.

\textbf{(3)\,Privacy}: Corruption of at most $k-1$ participants cannot acquire any information of secrets.

\subsection{Correctness}

\qquad At first, in order to test the consistency of $P_{i}'s$ subshadow $u_{i-1}$ with the released messages, we must publish the predefined constant $c$. When $c$ is unknown, the right side of the following equation is indeterminate.

$$ \prod_{j=0}^{k}(T_{i+1+k-j})^{\left( {\begin{array}{*{20}{ccc}}
	k\\
	j
	\end{array}} \right)}\overset{?}{=}g^{(-1)^{i}c}\:mod\:q.$$

Secondly, we explain the reason why we employ the $[NLR1]$ in Scheme 1. If we still use the original $NLR$ sequence \cite{MASHHADI201531}, which is
$$\left\{
\begin{aligned}
& u_{0}=c_{0},u_{1}=c_{1},\cdots,u_{k-2}=c_{k-2},\\
& \sum_{j=1}^{k}\left( {\begin{array}{*{20}{ccc}}
	k-1\\
	j-1
	\end{array}} \right)u_{n+k-j}=(-1)^{n}c\:mod \:Q \quad (n\geq 0).
\end{aligned}
\right.$$
We suppose that the corrupted $k-1$ participants are $\{P_{1},P_{2},\cdots,P_{k-1}\}$ with corresponding subshadows $\{u_{0},u_{1},\cdots,u_{k-2}\}$. Then the attacker can calculate the whole $NLR$ sequence since $c$ is released by the dealer $D$. In other words, an attacker corrupting $k-1$ participants can get other honest participants' subshadows, and the attacker can reconstruct the shared secrets successfully.

Thirdly, since we use $[NLR1]$ in the construction phase, we need construct a polynomial $p(x)$ with degree $k$ in the reconstruction phase as shown in the Theorem 1. There are $k+1$ indeterminate coefficients in $p(x)$, but we merely have $k$ subshadows. Besides the $k$ subshadows provided by $k$ honest participants, we need one more $u_{j}$. Because $\{u_{0},u_{1},\cdots,u_{m-1}\}$ are subshadows of participants $P_{1},P_{2},\cdots,P_{m}$, and $\{u_{m},u_{m+1},\cdots,u_{m+l-1}\}$ are correlated with the shared secrets $S_{1},S_{2},\cdots,S_{l}$. Then $u_{m+l}$ is the only term satisfying the demand. Therefore, we release $u_{m+l}$ in the construction phase.

Finally, if the dealer and the participants are honest, any at least $k$ participants can reconstruct the shared secrets using the two methods mentioned in the Section 3.1.4.
\begin{remark}
In fact, the aforementioned $NLR$ sequence has degree $k-1$. In MS schemes \cite{MASHHADI201531}, the constant $c$ is not released. Therefore, the designer must use a $k-1$-degree $NLR$ to satisfy the requirement of a secure $(k,m)$ VMSS scheme.

To sum up, we employ the $[NLR1]$ for the following reasons:

(1) The constant $c$ has to be released for verifying the consistency of $P_{i}'s$ subshadow $u_{i-1}$ with released messages.

(2) A $k$-degree $NLR$ sequence is utilized and the term $u_{m+l}$ has to be released to make Scheme 1 be a secure $(k,m)$ VMSS scheme.

\end{remark}
%%%%%%%%%%%%%%%%%%%%%%%%%%%%%%%%%%%%%%%%%%%%%%%%%%%%%%%%%%%%%%%%%%%%%%%%%%%
\subsection{Verifiability}

\qquad In the verification phase, if $u_{i-1}$ is valid, and $D$ succeeds in providing a false $u_{i-1}'$ to $P_{i}(u_{i-1}\neq u_{i-1}')$, which means that  $T_{i}=g^{u_{i-1}}=g^{u_{i-1}'}\: mod \:q$, and
$$ \prod_{j=0}^{k}(T_{i+1+k-j})^{\left( {\begin{array}{*{20}{ccc}}
	k\\
	j
	\end{array}} \right)}=g^{(-1)^{i}c}\:mod\:q.$$
Because $u_{i-1},u_{i-1}'\in Z_{Q}$, the possibility of $u_{i-1}\neq u_{i-1}'$ can be neglected. Then we conclude that $u_{i-1}'= u_{i-1}$, which means that $D$ cannot distribute a false $u_{i-1}'$ to $P_{i}$ successfully.

Besides, once a malicious $P_{i}$ succeeds in providing a false $u_{i-1}'$ during the reconstruction phase, the other participants get $T_{i}=g^{u_{i-1}}=g^{u_{i-1}'}\: mod \:q$, which implies $u_{i-1}'=u_{i-1}\:mod \:Q$. Then the dishonest $P_{i}$ can be discovered.

%%%%%%%%%%%%%%%%%%%%%%%%%%%%%%%%%%%%%%%%%%%%%%%%%%%%%%%%%%%%%%%%%%%%%%%%%%%
\subsection{Privacy}

\qquad Because $T_{i}=g^{u_{i-1}} \; mod \; q(i=1,2,\cdots,m)$ and the discrete logarithm problem in $Z_{q}^{\ast}$ with the base $g$ is hard, the attacker gets no helpful messages of $u_{i-1}$ from $T_{1},T_{2},\cdots,T_{m}$. If the attacker wants to obtain messages of $u_{i-1}$ from $H_{i}$, where $H_{i}=s_{e_{i}}(u_{i-1},u_{i-1}) \; mod \;N_{i}$ for $i=1,2,\cdots,m$, he must break the LFSR public key cryptosystem, which is impossible under our assumption.

Not mastering at least $k$ subshadows, the attacker cannot utilize the $NLR$ sequence $(u_{j})_{j\geq 0}$ to get $u_{m},u_{m+1},\cdots,u_{m+l-1}$. Then, the attacker gets nothing about shared secrets, namely $u_{m}+y_{1},u_{m+1}+y_{2},\cdots,u_{m+l-1}+y_{l}$.

\begin{theorem}%Theorem 3
The corruption of at most $k-1$ participants cannot obtain any helpful messages of the secrets.
\end{theorem}

\begin{proof}
The attacker cannot obtain any helpful messages of the $[NLR1]$ from $T_{i}$ and $H_{i}$. Therefore, in order to obtain honest participants' subshadows, the attacker can only utilize the at most $k-1$ corrupted participants' subshadows. We might as well assume $\{P_{1},P_{2},\cdots,P_{k-1}\}$ to be the malicious participants. The attacker can merely get the $NLR$ sequence $(u_{n})_{n\geq 0}$ from the formulas below by $\{u_{0},u_{1},\cdots,u_{k-2}\}$ and released constant $c$:

$$\left\{
\begin{aligned}
& u_{0}=c_{0},u_{1}=c_{1},\cdots,u_{k-2}=c_{k-2},\\
& \sum_{j=0}^{k}\left( {\begin{array}{*{20}{ccc}}
	k\\
	j
	\end{array}} \right)u_{n+k-j}=(-1)^{n}c\:mod \:Q,\quad n=0,1,\cdots,m-k-1.
\end{aligned}
\right.$$
There are $m-k$ linear equations, but $m-k+1$ variables $u_{k-1},u_{k},\cdots,u_{m-1}$. Then the attacker has to select the value of $u_{k-1}$ in $Z_{Q}$ randomly, which implies the probability to get other participants' subshadows successfully is $1/Q$. Since $Q>2^{\frac{\lambda}{2}}$, the probability is less than $1/2^{\frac{\lambda}{2}}$. Therefore, the corruption of at most $k-1$ participants cannot obtain any helpful messages of the secrets.

\end{proof}

%%%%%%%%%%%%%%%%%%%%%%%%%%%%%%%%%%%%%%%%%%%%%%%%%%%%%%%%%%%%%%%%%%%%%%%%%%%%

\section{Performance analysis}
\qquad In this section, we compare some proposed VMSS schemes with ours from four aspects.

\subsection{Computational cost}

\qquad As for the computational cost, since the modular exponentiation costs the great amount of time, Table 1 displays the difference in some presented schemes and ours. And Table 2 demonstrates the cryptographic knowledge used in the four phases of these schemes.

\begin{table}[!htbp]%Table 1
\caption{\textbf{$Computational\,cost$}}
\scalebox{1.055}[1.1]{
\begin{tabular}{|c|c|c|c|c|c|}
\hline
 \multirow{2}{*}{Scheme} &
 \multicolumn{2}{c|}{Initialization} &
 \multicolumn{1}{c|}{Construction} &
 \multicolumn{1}{c|}{Verification} &
 \multicolumn{1}{c|}{Reconstruction}\\
 \cline{2-6}
   &  $D$ & $P_{i}$ & $D$ & $P_{i}$ &  $P_{i}$ \\
 \hline
 MS1\cite{HADIANDEHKORDI20082262}  & $m$ & 1 & $m$   & - & $k-1$\\
 \hline
 MS2\cite{HADIANDEHKORDI20082262}  & 0 & 1 & $m+1$ & - & $k$\\
 \hline
 LZZ1\cite{Liu:2016:AVM:2869973.2870260} & 0 & 0 & $m+k(l\leq k)$ & $k+2(l\leq k)$ & $(k-1)(l+1)$\\
 \hline
 LZZ2\cite{Liu:2016:AVM:2869973.2870260} & 0 & 0 & $2m$ & $k+3$ & $k-1$\\
 \hline
 our schemes & 0 & 0 & $m$ & $k+3$ & $k-1$\\
 \hline
\end{tabular}}
\end{table}

\begin{table}[h]%Table 2
\caption{Cryptographic knowledge used in cited schemes }
\scalebox{0.85}[0.8]{
\begin{tabular}{|c|c|c|c|c|}
\hline
Scheme & Initialization & Construction & Verification & Reconstruction\\
\hline
  \makecell*[c]{MS1\cite{HADIANDEHKORDI20082262}\\MS2\cite{HADIANDEHKORDI20082262}} &
  \makecell*[c]{RSA\\Diffie-Hellman}&
  \makecell*[c]{HLR \\RSA} &
  \makecell*[c]{RSA\\Diffie-Hellman}&
  \makecell*[c]{HLR or \\ Lagrange \\ Interpolation}\\
\hline
  \makecell*[c]{LZZ1\cite{Liu:2016:AVM:2869973.2870260}} &
  \makecell*[c]{RSA}&
  \makecell*[c] {$(l-1)$-degree polynomial or \\$(k-1)$-degree polynomial\\RSA } &
  \makecell*[c]{RSA\\Diffie-Hellman}&
  \makecell*[c]{Lagrange\\ Interpolation}\\
\hline
  \makecell*[c]{LZZ2\cite{Liu:2016:AVM:2869973.2870260}} &
  \makecell*[c]{RSA}&
  \makecell*[c]{HLR \\RSA} &
  \makecell*[c]{RSA\\Diffie-Hellman}&
  \makecell*[c]{HLR or \\ Lagrange \\Interpolation}\\
\hline
  \makecell*[c]{MS\cite{MASHHADI201531}} &
  \makecell*[c]{LFSR PK \\ cryptosystem}&
  \makecell*[c]{NLR \\LFSR PK cryptosystem} &
  \makecell*[c]{LFSR PK \\ cryptosystem} &
  \makecell*[c]{NLR or \\ Lagrange \\Interpolation}\\
\hline
  \makecell*[c]{our schemes} &
  \makecell*[c]{LFSR PK \\ cryptosystem} &
  \makecell*[c]{NLR \\LFSR PK cryptosystem} &
  \makecell*[c]{LFSR PK \\cryptosystem\\Diffie-Hellman} &
  \makecell*[c]{NLR or \\ Lagrange\\ Interpolation}\\
\hline
\end{tabular}}
\end{table}

From Table 1, we know that our schemes are the most effective in first two phases. Compared with the schemes in \cite{Liu:2016:AVM:2869973.2870260}, our schemes replace the RSA encryption scheme by LFSR public key cryptosystem in the construction phase, which can be seen from Table 2. This replacement decreases the number of modular exponentiation used in the construction phase of our schemes. Since we append consistency check to detect the malicious dealer, our schemes need more modular exponentiations than the first two schemes and MS schemes \cite{MASHHADI201531} in the verification phase. Because the schemes in \cite{MASHHADI201531} do not utilize the modular exponentiation, we do not list it in Table 1.

%%%%%%%%%%%%%%%%%%%%%%%%%%%%%%%%%%%%%%%%%%%%%%%%%%%%%%%%%%%%%%%%%%%%%%%%%%%%
\subsection{Communication cost}

\qquad We give the communication cost in the first two phases of some schemes in Table 3. It implies that new VMSS schemes are as effective as MS1, MS2, MS schemes in the first phase, but a little less effective in the second phase, owing to the fact that we request $D$ to provide released messages for detecting the malicious behavior of the dealer. The serious consequences of this shortage is showed in the Section 2.3.

\begin{table}[h]% Table 3
\caption{\textbf{$Communication\,cost$}}
\scalebox{0.75}[0.85]{
\begin{tabular}{|c|c|c|c|}
\hline
 \multirow{2}{*}{Scheme} &
 \multicolumn{2}{c|}{Initialization} &
 \multicolumn{1}{c|}{Construction} \\
\cline{2-4}
 & $D$  Publish & $P_{i}$ to $D$ & $D$ Publish\\
\hline
  \makecell*[c]{MS1\cite{HADIANDEHKORDI20082262}} &
  \makecell*[c]{$(e,N,g,q,\alpha)$\\$(ID_{i},T_{i}),i=1,2,\cdots,m$ } &
  \makecell*[c]{$(ID_{i},T_{i})$,\\$i=1,2,\cdots,m $} &
  \makecell*[c]{($r,G_{1},G_{2},\cdots,G_{m},r_{1},r_{2},$\\
  $\cdots,r_{l},y_{k+1},y_{k+2},\cdots,y_{m})$}\\
\hline
  \makecell*[c]{MS2\cite{HADIANDEHKORDI20082262}} &
  \makecell*[c]{$(N,g,q,\alpha)$\\$(i,R_{i}),i=1,2,\cdots,m$ } &
  \makecell*[c]{$(i,R_{i})$,\\$i=1,2,\cdots,m $} &
  \makecell*[c]{($R_{0},f,r_{1},r_{2},\cdots,r_{l}$,\\
  $y_{k+1},y_{k+2},\cdots,y_{m})$} \\
\hline
  \makecell*[c]{LZZ1\cite{Liu:2016:AVM:2869973.2870260}} &
  \makecell*[c]{$(\lambda,N,Q,q,g)$\\$(ID_{i},e_{i},N_{i}),i=1,2,\cdots,m$ } &
  \makecell*[c]{$(ID_{i},e_{i},N_{i})$,\\$i=1,2,\cdots,m $} &
  \makecell*[c]{$l\leq k,(C_{1},C_{2},\cdots,C_{m},H_{1},H_{2},$\\
  $\qquad\cdots,H_{m},A_{1},A_{2},\cdots,A_{k})$ \\
  $l>k,(C_{1},C_{2},\cdots,C_{m},H_{1},H_{2},\cdots,H_{m},$\\
  $\qquad\eta_{1},\eta_{2},\cdots,\eta_{l-k},A_{1},A_{2},\cdots,A_{l},$\\
  $f(\eta_{1}),f(\eta_{2}),\cdots,f(\eta_{l-k}))$}\\
\hline
  \makecell*[c]{LZZ2\cite{Liu:2016:AVM:2869973.2870260}} &
  \makecell*[c]{$(\lambda,N,Q,q,g,\alpha)$\\$(ID_{i},e_{i},N_{i}),i=1,2,\cdots,m$ } &
  \makecell*[c]{$(ID_{i},e_{i},N_{i})$,\\$i=1,2,\cdots,m$} &
  \makecell*[c]{$(H_{1},H_{2},\cdots,H_{m},T_{1},T_{2},\cdots,T_{m},$\\
  $Y_{1},Y_{2},\cdots,Y_{l})$}\\
\hline
  \makecell*[c]{MS\cite{MASHHADI201531}} &
  \makecell*[c]{$(N,a,b,q_{1})$\\$(ID_{i},s_{e_{i}}(a,b)),i=1,2,\cdots,m$ } &
  \makecell*[c]{$(ID_{i},s_{e_{i}}(a,b))$,\\$i=1,2,\cdots,m $} &
  \makecell*[c]{$(s_{e_{0}}(a,b),d,r_{1},r_{2},\cdots,r_{l},y_{k},y_{k+1},\cdots,y_{m}$)}\\
\hline
  \makecell*[c]{our schemes} &
  \makecell*[c]{$(\lambda,N,Q,q,g)$\\$(ID_{i},e_{i},N_{i}),i=1,2,\cdots,m$ } &
  \makecell*[c]{$(ID_{i},e_{i},N_{i})$,\\$i=1,2,\cdots,m $} &
  \makecell*[c]{$(H_{1},H_{2},\cdots,H_{m},T_{1},T_{2},\cdots,T_{m},$\\
  $y_{1},y_{2},\cdots,y_{l},c,u_{m+l}$)}\\
\hline
\end{tabular}}
\end{table}

%%%%%%%%%%%%%%%%%%%%%%%%%%%%%%%%%%%%%%%%%%%%%%%%%%%%%%%%%%%%%%%%%%%%%%%%%%%
\subsection{Dynamism}

\qquad In this subsection, we will illustrate how to process a dynamic update, cancel, and addition of the participants, the values of secrets and the threshold according to the actual situation.

\noindent{\textbf{Participants:}}

If a fresh participant $P_{new}$ wants to participate in the scheme, it computes $N_{new}=p_{new}q_{new}$, where $p_{new}$ and $q_{new}$ are strong primes. Then it chooses an integer $e_{new}$ and calculates the corresponding secret shadow $d_{new}$, then sends $(ID_{new},e_{new},N_{new})$ to $D$. Next $D$ calculates $H_{new}=s_{e_{new}}(u_{new-1},u_{new-1})\:mod\:N_{new}$ and $T_{new}=g^{u_{new-1}}\:mod\:q$, then releases them. When it comes to canceling a $P_{can}$, $D$ merely eliminates $(ID_{can},e_{can},N_{can})$. Once $P_{can}$ tries to utilize $u_{can-1}$ to recover the secrets, it is impossible not to be detected.

\noindent{\textbf{Secrets:}}

If $D$ wants to append a secret $S_{l+1}$, he merely calculates and releases $y_{l+1}=S_{l+1}-u_{m+(l+1)-1}$. Similarly, if $D$ wants to reduce a secret $S_{i}$, he merely eliminates $y_{i}=S_{i}-u_{m+i-1}$. When $D$ wants to alter the value of the secrets, the manipulation is completely evident.

\noindent{\textbf{Threshold:}}

In addition, as we analyzed in Section 4.1, our schemes are secure $(k,m)$ VMSS schemes. If $D$ wants to change the threshold of our schemes, he only need utilize a new $NLR$ sequence with a corresponding degree. For example, if we use the $NLR$ sequence with $k-1$ degree, our schemes are secure $(k-1,m)$ VMSS schemes, which is executable easily. Since the schemes in \cite{MASHHADI201531} are $(k,m)$ VMSS schemes, in order to compare them with ours, we also require that our schemes are secure $(k,m)$ VMSS schemes.

\subsection{Performance characteristic}
\qquad We analyze the performance characteristic of the schemes in \cite{HADIANDEHKORDI20082262,Liu:2016:AVM:2869973.2870260,MASHHADI201531} and our schemes in Table 4.

$\bullet$ Characteristic 1: Recover multiple secrets simultaneously

$\bullet$ Characteristic 2: Usage of the public channel

$\bullet$ Characteristic 3: Detect deception of malicious $D$

$\bullet$ Characteristic 4: Detect deception of malicious $P_{i}$

$\bullet$ Characteristic 5: Change the shared secrets after an unsuccessful reconstruction phase

$\bullet$ Characteristic 6: Recycle of the shadows with diverse $D$

$\bullet$ Characteristic 7: Recycle of the shadows when participants join in/quit from the group

$\bullet$ Characteristic 8: The bit length of private key in a 1024-bit finite field

$\bullet$ Characteristic 9: The bit length of public key in a 1024-bit finite field

\begin{table}[!htbp]%Table 4
\caption{$Performance\, characteristic$}
\scalebox{1.16}[0.9]{
\begin{tabular}{|c|c|c|c|c|c|}
\hline
Characteristic  & MS1\cite{HADIANDEHKORDI20082262} & MS2\cite{HADIANDEHKORDI20082262} & LZZ1,2\cite{Liu:2016:AVM:2869973.2870260} &  MS\cite{MASHHADI201531} & our schemes\\
\hline
1  & YES & YES & YES & YES & YES\\
\hline
2 & YES & YES & YES & YES & YES\\
\hline
3 & NO & NO & YES & NO & YES \\
\hline
4  & YES & YES & YES & YES & YES\\
\hline
5  & NO & NO & NO & NO & NO\\
\hline
6  & NO & YES & YES & YES & YES\\
\hline
7  & YES & YES & YES & YES & YES\\
\hline
8  & 1024 & 1024 & 1024 & 340 & 340 \\
\hline
9 & 1024 & 1024 & 1024 & 340 & 340 \\
\hline
\end{tabular}}
\end{table}

From Table 4, MS1, MS2 and MS schemes cannot detect deception by malicious $D$, while LZZ1, LZZ2 and our schemes can overcome the drawback. However, in a $1024$-bit finite field, the length of the private or public key in our schemes are denoted by only $340$ bits, while in \cite{Liu:2016:AVM:2869973.2870260} the length is three times longer to achieve the same safety level, which implies our schemes are more efficient and have lower consumption. Besides, in MS1 \cite{HADIANDEHKORDI20082262} scheme, the shadows of participants cannot be reused after reconstructing the secrets. Therefore, the participants have to operate the first phase repeatedly, while our schemes allow them to reuse shadows, which reduces the consumption of initialization.

%%%%%%%%%%%%%%%%%%%%%%%%%%%%%%%%%%%%%%%%%%%%%%%%%%%%%%%%%%%%%%%%%%%%%%%%%
\section{Conclusion}

\qquad Dynamic and verifiable multi-secret sharing schemes share multiple secrets among a set of participants and detect the deception by malicious participants and the dealer dynamically. Utilizing the nonhomogeneous linear recursion and LFSR public key cryptosystem, we propose two efficient dynamic and verifiable multi-secret sharing schemes.

First, our schemes conquer the drawback of MS schemes \cite{MASHHADI201531} by adding consistency check between the participants' corresponding subshadows and released messages. Second, although our schemes have the same advantage as the schemes in \cite{Liu:2016:AVM:2869973.2870260}, we have less computational cost. Furthermore, since we substitute the RSA encryption scheme by the LFSR public key cryptosystem, the private/public key length of our schemes is only one-third of the schemes in \cite{Liu:2016:AVM:2869973.2870260} for the same safety level.

The final analyses of the proposed schemes indicate that they are secure and effective $(k,m)$ VMSS schemes, permitting recovery of multiple secrets simultaneously, using the public channel, detecting deception of both a malicious dealer and participants, reusing the shadows, having dynamism attribute, and having shorter public/private key length.

\bibliographystyle{plain}%
\bibliography{bibfile}

\end{document}